\documentclass[11pt,draftcls,onecolumn]{IEEEtran}
\IEEEoverridecommandlockouts                              % This command is only
\usepackage{amssymb,latexsym,color,amsmath,pifont,colordvi,multicol}
\usepackage{times}
\definecolor{backgrey}{rgb}{0.86,0.86,0.86}
\definecolor{dblue}{rgb}{0,0.0,0.5}
\definecolor{dred}{rgb}{0.4,0.2,0}
\definecolor{dgreen}{rgb}{0.0,0.5,0}

\newcommand{\captionfonts}{\small}
\makeatletter  % Allow the use of @ in command names
\long\def\@makecaption#1#2{%
  \vskip\abovecaptionskip
  \sbox\@tempboxa{{\captionfonts #1: #2}}%
  \ifdim \wd\@tempboxa >\hsize
    {\captionfonts #1: #2\par}
  \else
    \hbox to\hsize{\hfil\box\@tempboxa\hfil}%
  \fi
  \vskip\belowcaptionskip}
\makeatother   % Cancel the effect of \makeatletter

\usepackage{graphics} % for pdf, bitmapped graphics files
\usepackage[pdftex]{graphicx}

\newtheorem{theorem}{Theorem}

\newtheorem{proof}{Proof}

\newtheorem{definition}{Definition}

\newtheorem{lemma}{Lemma}

\title{\LARGE \bf Information Flow Decomposition in Feedback Systems: Linear Time-Invariant Systems with Gaussian Channels}

\author{\quad Bertrand Wechsler, Dan Eilat and Nicolas Limal
}

\begin{document}
\maketitle \thispagestyle{empty} \pagestyle{plain}
\begin{abstract}
In our companion paper \cite{Limal_decomposition}, an information identity decomposition has been derived, which can be interpreted as a law of conservation of information flows in feedback systems. In this paper, we further investigate this decomposition result when specified to linear time-invariant(LTI) systems connected with additive white Gaussian noise(AWGN) channels. It is shown that the quantities in the decomposition are characterized in sensitivity function and the law of conservation is verified.
\end{abstract}

\section{Introduction}
\indent Consider a feedback system as shown in Fig. \ref{Fig:general_feedback_system} (left). Let $m = x_0$, the information flow through channel $C_2$ can be decomposed into two independent flows as
\begin{equation*}
I(y^{n}\rightarrow e^n) = I(x^{n}\rightarrow e^n) + I(y^n\rightarrow e^n|x_0).
\end{equation*}

\begin{figure}
\begin{center}
\includegraphics[scale=0.70]{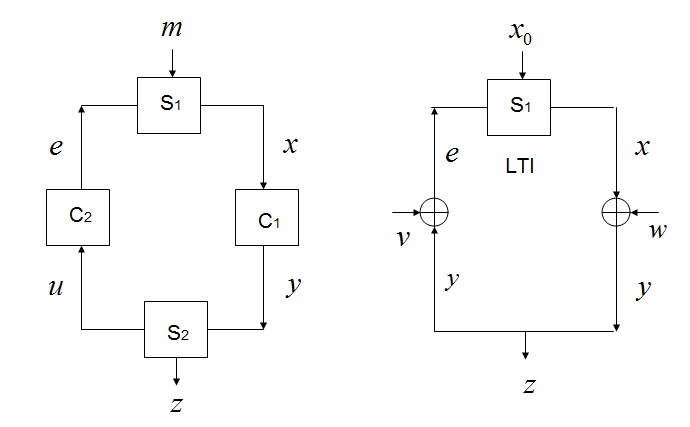}
\caption{Left: a general feedback system; Right: LTI systems connected with AWGN channels.}
\label{Fig:general_feedback_system}
\end{center}
\end{figure}

This equality can be interpreted as a law of conservation of information flows. Quantity $I(x^{n}\rightarrow e^n)$ is the amount of information provided by the external input $x_0$, and the  quantity $I(y^n\rightarrow e^n|x_0)$ is the amount of information provided by the uncertainty in the channel $C_1$ (due to the presence of noise). The sum of these two quantities equals to the total amount of information delivered from system $S_2$ to system $S_1$ through channel $C_2$.\\
\indent In this paper, we turn our attention to Linear-time-invariant (LTI) systems connected with AWGN channels, as shown in Fig. \ref{Fig:general_feedback_system} (right), where system $S_1$ and $S_2$ are respectively assumed to be LTI and unit gain and $(w,v)$ are AWGNs with variance $(\sigma_v^2,\sigma_w^2 )$.  As it will be shown, the limit values of these quantities in the above information flow equation can be characterized in sensitivity function and, furthermore, verify the law of conservation of information flow. Although the topic of distributed systems connected with one noisy channel and one noiseless channel (e.g. \cite{Kim08_capacity_fb, Kim10,Ofer07,Tati09,Permuter09}) or both noisy channels (e.g., \cite{Draper06,chong_arXiv11, chong_isit11, Kim07,Chong11_isit_bounds,Chong11_allerton_UpperBound,Martins08, Chong11_allerton_finiteCapacity, Chance10, Chong12_allerton_sideInfo, Limal_infoIdentity}) is of interest and has attracted much attention in the last decade, most of the work has concentrated on communication systems with feedback, where the coding schemes are not necessarily LTI system. The result of this paper contributes to the understanding of noisy feedback systems in the control-theoretic perspective, where the systems are mostly restricted to be LTI.\\

\section{Decomposition on LTI Feedback Systems with Gaussian Channels}

\indent First of all, we revisit the definition of directed information which will be repeatedly used in the paper.
\begin{definition}
Given random sequences $x^n$, $y^n$, the directed information from $x^n$ to $y^n$ is defined as
\begin{equation*}
I(x^n \rightarrow y^n) = \sum_{i=1}^{n}I(x^i;y_i|y^{i-1}).
\end{equation*}
\end{definition}

\begin{definition}
Consider feedback systems as shown in \ref{Fig:general_feedback_system}, the sensitivity transfer function $\mathcal{S}(e^{j2\pi\theta})$ is defined as the transfer function from the external disturbances $w$ to the process output $y$ or the measurement noise $v$ to the system inputs $e$.
\end{definition}
The sensitivity function is introduced and widely used in control theory. According to the definition, lower values of $|\mathcal{S}(e^{j2\pi\theta})|$ suggest further attenuation of the external disturbances or the measurement noise. Moreover, the sensitivity function also reflects the feedback influence on external disturbances.

%Denote $\mathcal{S}_x(e^{j2\pi\theta})$ as frequency-domain representation of the stationary signal $\lbrace X_i\rbrace_{i=-\infty}^{\infty}$.
\begin{lemma}
Consider a stationary Gaussian process $\lbrace X_i\rbrace_{i=1}^{\infty}$, the relationship between its entropy rate $\bar{h}(X) = \lim_{n\rightarrow \infty} \frac{1}{n} h(X^n)$ and the spectral density $S_x(e^{j\theta})$ is in the following expression,
\begin{equation*}
\bar{h}(X) = \frac{1}{2}\int_{-\frac{1}{2}}^{\frac{1}{2}} \ln{2\pi e S_x(e^{j\theta})}d\theta
\end{equation*}
\label{entropy_rate}
\end{lemma}

For convenience, we again recall the decomposition theorem (law of conservation of information flows) in \cite{Limal_decomposition} as follows. Note that this decomposition is essentially introduced in \cite{derpich2013fundamental} under a general framework.
\begin{theorem}
Consider a feedback system as shown in Fig. \ref{Fig:general_feedback_system} (left). Let $m = x_0$, the information flow through channel $C_2$ can be decomposed into two independent flows as
\begin{equation*}
I(y^{n}\rightarrow e^n) = I(x^{n}\rightarrow e^n) + I(y^n\rightarrow e^n|x_0).
\end{equation*}
\end{theorem}

In what follows, this result is applied to the model as shown in Fig. \ref{Fig:general_feedback_system}(right). All the three quantities are explicitly characterized in sensitivity function and the law of conservation is verified.

\begin{theorem}
Consider LTI feedback systems with AWGN channels as shown in Fig. \ref{Fig:general_feedback_system} (right). Then,
\begin{equation*}
\begin{split}
\lim_{n\rightarrow \infty}\frac{1}{n}I(y^{n}\rightarrow e^n)=&\int_{-\frac{1}{2}}^{\frac{1}{2}} \ln{|\mathcal{S}(e^{j2\pi \theta})|} d\theta + \frac{1}{2}\int_{-\frac{1}{2}}^{\frac{1}{2}} \ln (1 +\frac{\sigma_w^2}{\sigma_v^2}) d\theta\\
\lim_{n\rightarrow \infty}\frac{1}{n}I(x^{n}\rightarrow e^n)=&\int_{-\frac{1}{2}}^{\frac{1}{2}} \ln{|\mathcal{S}(e^{j2\pi \theta})|} d\theta \\
\lim_{n\rightarrow \infty}\frac{1}{n}I(y^n\rightarrow e^n|x_0)=&\frac{1}{2}\int_{-\frac{1}{2}}^{\frac{1}{2}} \ln (1 +\frac{\sigma_w^2}{\sigma_v^2}) d\theta\\
\end{split}
\end{equation*}
\end{theorem}

\begin{proof}
We prove the above three equalities one by one.
(1). First of all, we have
\begin{equation*}
\begin{split}
I(y^n\rightarrow e^n)=&\sum_{i=1}^{n}I(y^i, e_i|e^{i-1})\\
=&\sum_{i=1}^{n}h(e_i|e^{i-1})-h(e_i|e^{i-1}, y^i)\\
=&\sum_{i=1}^{n}h(e_i|e^{i-1})-h(y_i+ v_i|y_1+v_1, y_2+v_2,\cdots, y_{i-1}+v_{i-1}, y^i)\\
=&\sum_{i=1}^{n}h(e_i|e^{i-1})-h(v_i|v_1, v_2,\cdots,v_{i-1}, y^i)\\
=&\sum_{i=1}^{n}h(e_i|e^{i-1})-h(v_i|v^{i-1})\\
=&h(e^{n})-h(v^{n})\\
\end{split}
\end{equation*}
Then, according to Lemma \ref{entropy_rate},
\begin{equation*}
\begin{split}
\lim_{n\rightarrow \infty}\frac{1}{n}I(y^n\rightarrow e^n)=&\lim_{n\rightarrow \infty}\frac{1}{n}h(e^{n})-h(v^{n})\\
=&\frac{1}{2}\int_{-\frac{1}{2}}^{\frac{1}{2}} \ln{2\pi e S_e(e^{j\theta})}d\theta-\frac{1}{2}\int_{-\frac{1}{2}}^{\frac{1}{2}} \ln{2\pi e S_v(e^{j\theta})}d\theta\\
%=&\lim_{n\rightarrow \infty}\frac{1}{2n} \ln{\frac{|\Sigma_{e^n}|}{|\Sigma_{v^n}|}}\\
=& \frac{1}{2}\int_{-\frac{1}{2}}^{\frac{1}{2}} \ln{\frac{S_e(e^{j2\pi \theta})}{S_v(e^{j2\pi \theta})}}d\theta\\
=& \frac{1}{2}\int_{-\frac{1}{2}}^{\frac{1}{2}} \ln{\frac{|\mathcal{S}(e^{j2\pi \theta})|^2(S_v(e^{j2\pi \theta})+S_w(e^{j2\pi \theta})) }{S_v(e^{j2\pi \theta})}}d\theta\\
=& \frac{1}{2}\int_{-\frac{1}{2}}^{\frac{1}{2}} \ln{\frac{|\mathcal{S}(e^{j2\pi \theta})|^2(\sigma_v^2+\sigma_w^2)}{\sigma_v^2}}d\theta\\
=& \int_{-\frac{1}{2}}^{\frac{1}{2}} \ln{|\mathcal{S}(e^{j2\pi \theta})|} d\theta + \frac{1}{2}\int_{-\frac{1}{2}}^{\frac{1}{2}} \ln (1 +\frac{\sigma_w^2}{\sigma_v^2}) d\theta\\
\end{split}
\end{equation*}

(2). Next,
\begin{equation*}
\begin{split}
I(x^{n}\rightarrow e^n)=&\sum_{i=1}^{n}I(x^i, e_i|e^{i-1})\\
=&\sum_{i=1}^{n}h(e_i|e^{i-1})-h(e_i|e^{i-1}, x^i)\\
=&\sum_{i=1}^{n}h(e_i|e^{i-1})-h(x_i+ w_i + v_i|x_1+w_1+v_1, x_2+w_2+v_2,\cdots, x_{i-1}+w_{i-1}+v_{i-1}, x^i)\\
=&\sum_{i=1}^{n}h(e_i|e^{i-1})-h(w_i + v_i|w_1+v_1, w_2+v_2,\cdots,w_{i-1}+v_{i-1}, x^i)\\
=&\sum_{i=1}^{n}h(e_i|e^{i-1})-h(w_i + v_i|w^{i-1}+v^{i-1})\\
=&h(e^{n})-h(v^{n}+w^n)\\
\end{split}
\end{equation*}
Similarly, we have
\begin{equation*}
\begin{split}
\lim_{n\rightarrow \infty}\frac{1}{n}I(x^{n}\rightarrow e^n)=&\lim_{n\rightarrow \infty}\frac{1}{n}h(e^{n})-h(v^{n}+w^n)\\
=&\frac{1}{2}\int_{-\frac{1}{2}}^{\frac{1}{2}} \ln{2\pi e S_e(e^{j\theta})}d\theta-\frac{1}{2}\int_{-\frac{1}{2}}^{\frac{1}{2}} \ln{2\pi e S_w(e^{j\theta})+S_v(e^{j\theta})}d\theta\\
%=&\lim_{n\rightarrow \infty}\frac{1}{2n} \ln{\frac{|\Sigma_{e^n}|}{|\Sigma_{w^n+v^n}|}}\\
=& \frac{1}{2}\int_{-\frac{1}{2}}^{\frac{1}{2}} \ln{\frac{S_e(e^{j2\pi \theta})}{S_w(e^{j2\pi \theta})+S_v(e^{j2\pi \theta})}}d\theta\\
=& \frac{1}{2}\int_{-\frac{1}{2}}^{\frac{1}{2}} \ln{\frac{|\mathcal{S}(e^{j2\pi \theta})|^2(\sigma_v^2+\sigma_w^2)}{\sigma_v^2+\sigma_w^2}}d\theta\\
%=& \frac{1}{2}\int_{-\frac{1}{2}}^{\frac{1}{2}} \ln{\frac{|\mathcal{S}(e^{j2\pi \theta})|^2(\sigma_v^2+\sigma_w^2)}{(\sigma_v^2+\sigma_w^2)}}d\theta\\
=& \int_{-\frac{1}{2}}^{\frac{1}{2}} \ln{|\mathcal{S}(e^{j2\pi \theta})|} d\theta \\
\end{split}
\end{equation*}

(3). Finally, we need to show $\lim_{n\rightarrow \infty}\frac{1}{n}I(y^n\rightarrow e^n|x_0) = \int_{-\frac{1}{2}}^{\frac{1}{2}} \ln (1 +\frac{\sigma_w^2}{\sigma_v^2}) d\theta$.
\begin{equation*}
\begin{split}
I(y^n\rightarrow e^n|x_0)=&\sum_{i=1}^{n}I(y^i, e_i|e^{i-1}, x_0)\\
=&\sum_{i=1}^{n}h(e_i|e^{i-1}, x_0)-h(e_i|e^{i-1}, y^i, x_0)\\
=&\sum_{i=1}^{n}h(e_i|e^{i-1}, x_0)-h(y_i+ v_i|y_1+v_1, y_2+v_2,\cdots, y_{i-1}+v_{i-1}, y^i, x_0)\\
=&\sum_{i=1}^{n}h(e_i|e^{i-1}, x_0)-h(v_i|v_1, v_2,\cdots,v_{i-1}, y^i, x_0)\\
=&\sum_{i=1}^{n}h(e_i|e^{i-1}, x_0)-h(v_i|v^{i-1})\\
=&\sum_{i=1}^{n}h(x_i+ w_i + v_i|x_1+w_1+v_1, x_2+w_2+v_2,\cdots, x_{i-1}+w_{i-1}+v_{i-1}, x_0)-h(v_i|v^{i-1})\\
\end{split}
\end{equation*}
Because $S_1$ is assumed to be a causal LTI system $\lbrace g_i\rbrace_{i=1}^\infty$, the system output $x_i = g_i(x^{i-1},e^{i-1}) = g_i(x^{i-1},w^{i-1}+v^{i-1})$ and $x_1 = g_1(x_0)$. Thus,
\begin{equation*}
\begin{split}
I(y^n\rightarrow e^n|x_0)=&\sum_{i=1}^{n}h(x_i+ w_i + v_i|x_1+w_1+v_1, x_2+w_2+v_2,\cdots, x_{i-1}+w_{i-1}+v_{i-1}, x_0)-h(v_i|v^{i-1})\\
=&\sum_{i=1}^{n}h(x_i+ w_i + v_i|x_1, w_1+v_1, x_2+w_2+v_2,\cdots, x_{i-1}+w_{i-1}+v_{i-1}, x_0)-h(v_i|v^{i-1})\\
=&\sum_{i=1}^{n}h(x_i+ w_i + v_i|x_1, w_1+v_1, x_2, w_2+v_2,\cdots, x_{i-1}+w_{i-1}+v_{i-1}, x_0)-h(v_i|v^{i-1})\\
=&\sum_{i=1}^{n}h(w_i + v_i|w_1+v_1, w_2+v_2,\cdots,w_{i-1}+v_{i-1}, x^i)-h(v_i|v^{i-1})\\
=& h(w^n+v^n) - h(v^n)\\
\end{split}
\end{equation*}

Then, we have
\begin{equation*}
\begin{split}
\lim_{n\rightarrow \infty}\frac{1}{n}I(y^n\rightarrow e^n|x_0) = \lim_{n\rightarrow \infty}\frac{1}{n}h(w^n+v^n) - h(v^n)=\frac{1}{2}\int_{-\frac{1}{2}}^{\frac{1}{2}} \ln (1 +\frac{\sigma_w^2}{\sigma_v^2}) d\theta\\
\end{split}
\end{equation*}
\end{proof}

\section{Conclusion}
We characterized the quantities in the law of conservation of information flows for distributed LTI systems connected with AWGN channels. The characterizations are represented in sensitivity function and verifies the law of conservation.
\bibliographystyle{IEEEtran}
\bibliography{ref}

\end{document}